\documentclass[11pt]{article}

\usepackage[square,numbers,sort&compress]{natbib}
\usepackage{setspace}
\usepackage[margin=1in]{geometry}

\usepackage[T1]{fontenc}    
\usepackage[colorlinks=true, allcolors=blue]{hyperref}
\usepackage{hyperref}
\usepackage{url}            
\usepackage{booktabs}       
\usepackage{amsfonts}       
\usepackage{amsmath}
\usepackage{amssymb}
\usepackage{amsthm}
\usepackage{nicefrac}       
\usepackage{microtype}      
\usepackage{xcolor}         
\usepackage{enumitem}
\usepackage{bbold}
\usepackage{graphicx}
\usepackage{bbm}

\usepackage{booktabs}
\usepackage{caption}
\usepackage{subcaption}

\theoremstyle{plain}

\newtheorem{assumption}{Assumption}
\newtheorem{remark}{Remark}
\newtheorem{theorem}{Theorem}
\newtheorem{proposition}[theorem]{Proposition}
\newtheorem{lemma}[theorem]{Lemma}
\newtheorem{corollary}[theorem]{Corollary}

\newcommand{\E}{\mathbb{E}}

\newcommand{\X}{\mathcal X}
\newcommand{\F}{\mathcal F}
\newcommand{\C}{\mathcal C}
\newcommand{\N}{\mathbb N_0}
\newcommand{\Prob}{\mathbb P}
\newcommand{\one}{\mathbbm 1}

\title{Defensive Sufficiency in a Stackelberg Model of AI Security}

\author{
Subhabrata Majumdar$^1$\footnotemark[1] \quad
Rajlakshmi Chavan$^1$ \\[3pt]
\shortstack{
\small{$^1$Indian Institute of Management Bangalore}
}
}
\date{}

\begin{document}

\renewcommand{\thefootnote}{\fnsymbol{footnote}}
\maketitle

\footnotetext[1]{Corresponding author, email: \texttt{smajumdar@iimb.ac.in}\\
Supported by Indian Institute of Management Bangalore Young Faculty Research Grant.}
\renewcommand{\thefootnote}{\arabic{footnote}}

\begin{abstract}
Feedback from automated testing, human red teaming, and incident response can strengthen an AI system's defenses when discovered failures lead to effective repairs. We study when this feedback process provides sufficient protection and when investing in it is economically worthwhile. We begin by showing that an attack surface composed of   finite number of inputs is defended with probability 1 if every unresolved attack has a  persistent chance of discovery, repairs are effective, and subsequent updates preserve earlier protection. We derive completion-time bounds and extend the analysis to growing attack surfaces, repairs that generalize across related attacks, and multiple discovery mechanisms. These results distinguish eventual protection against each fixed attack from
complete protection at a single time. We then formulate a defender-led Stackelberg game in which the defender invests in proactive discovery and reactive repair, anticipating the attacker's choice of search effort. We characterize the least-cost allocation that deters attack and the equilibrium regimes in which the defender funds neither capability, one capability, or both. Numerical experiments illustrate these regimes and show how faster repair can reduce compromise duration without reducing compromise probability.
\end{abstract}

\section{Introduction}
\label{sec:intro}

Compared to traditional software, generative AI systems have an expanded attack surface. Alongside vulnerabilities in connected tools, software, infrastructure, and access control, defenders must deal with attacks that exploit how the large language model (LLM) inside the AI system interprets and acts on input information \citep{vassilev2025taxonomy}. For example, consider an AI assistant that reads
documents and uses tools on a user's behalf. A document can contain both information relevant to the task and instructions intended to redirect the assistant. Processing that document can therefore expose the system to an attack, even when the user's request is
benign. Indirect prompt injection demonstrates this risk in tool-calling LLM agents, including attempts to induce unauthorized actions and disclose private information
\citep{zhan2024injecagent}.

This broader attack surface changes what defenders must account for when designing and implementing protection. An attack instance can depend on the wording of a prompt, the preceding conversation, retrieved content, and the tools or permissions available to the model. Increasing a system's capabilities can also introduce new ways to attack it. Many-shot jailbreaking, for example, exploits longer context windows by supplying demonstrations that steer a model toward behavior its safeguards are intended to prevent
\citep{anil2024many,russinovich2025crescendo,mehrotra2024tap}. Thus, protection established for short, isolated prompts may not be enough for longer interactions or a deployment with tools enabled.

While the familiar security practice of discovering vulnerabilities, applying repairs, and verifying updates remains necessary to secure LLM-driven deployments \citep{souppaya2022patch}, the expanded attack surface requires us to examine more carefully what it takes to make these steps sufficient. Blocking one malicious prompt may still leave its paraphrases or similar prompts as effective attack vectors. A repair may involve input or output filters, safety training, retrieval controls, tool permissions, or application logic---its effectiveness depending on how these components interact inside the AI system. Subsequent changes can also weaken earlier protection; for
example, further fine-tuning can disrupt safety alignment \citep{qi2025safety}. On top of traditional software failures from insufficient testing and regression, LLM-driven systems add a broad space of context-dependent behaviors over which the defender must assess the repair of vulnerabilities found and the retention of those repairs.

Automated testing, human red teaming, and incident response can expose failures that inform later defenses through feedback. Methods such as MART combine repeated attack generation with safety training \citep{ge2024mart}, while a lifecycle view of AI red teaming emphasizes evaluation throughout development and deployment \citep{majumdar2025red}. These activities motivate
a research question: 

\begin{quote}
{\it Under what conditions does continued discovery and repair provide sufficient protection?}
\end{quote}

\noindent We study this question by representing a defense through the set of attacks against which it is effective. This separates three requirements: testing must continue to reach unresolved attacks, reported vulnerabilities must be repaired, and earlier
protection must be retained. Under explicit conditions on these requirements, we establish almost-sure complete coverage of a finite attack surface and bound the time required. We extend the analysis to growing attack surfaces, repairs that cover families of related attacks, and multiple discovery mechanisms.

The distinction between finite and growing attack surfaces is central to interpreting defensive sufficiency. On a countably infinite surface, every fixed attack can eventually be covered even though some attack remains uncovered at every finite time. A persistent attacker may continue to seek those remaining opportunities. Thus, eventual protection against each attack does not by itself establish that a deployed system faces
negligible current risk.

This leads to our second research question:

\begin{quote}
{\it How much should a defender invest in discovery and repair to deter an adaptive attacker?}
\end{quote}

\noindent A defender can fund proactive discovery to prevent compromise, or reactive repair to shorten its duration. Both choices affect the return an attacker expects from searching for a vulnerability. We formulate a defender-led Stackelberg game in which the defender commits to these capabilities before the attacker chooses its search effort. Our model characterizes the optimal investment regimes and the least-cost allocation that deters attack. Numerical experiments illustrate how the relative costs of attack and defense shape these decisions, and why reducing compromise duration can be valuable even when compromise remains possible. Finally, we emphasize the need for verification mechanisms to ensure the continued validity of discovery and repair. As an illustration, we implement a machine-checkable verification layer in Lean on top of the garak LLM scanner \cite{garak} and share it through an open-source repository\footnote{\url{https://github.com/rajlakshmichavan/garak-in-lean}}.


\section{Related Work}
\label{sec:related-work}

Our work connects two questions relevant to a defender: when repeated testing and repair
provide defensive coverage to an AI system, and how a defender should invest in the
capabilities that support this process. The first relates to red teaming, safety training, and defense guarantees. The second relates to security investment and attacker--defender games.

\paragraph{Finding failures and learning from them.}
Automated red teaming uses a model to generate inputs that expose failures in another model. \citet{perez2022red} study several ways to generate such tests and examine their effectiveness and diversity, \citet{hong2024curiosity} introduce curiosity-driven exploration to improve the coverage of generated attacks, and garak \cite{garak} and PyRIT \cite{pyrit} are two popular open-source LLM red teaming tools. These approaches address an important component of our framework: whether testing continues to reach vulnerabilities
that remain unresolved. Greater empirical diversity, however, does not establish the persistent discovery condition used in our coverage results.

MART~\citep{ge2024mart} and rainbow teaming \cite{rainbowteaming} make the connection between discovery and repair explicit by alternating adversarial prompt generation with safety fine-tuning of the target model. HarmBench~\citep{mazeika2024harmbench} provides an evaluation framework and an adversarial training method to improve refusal robustness. It is also easy to update input or output guardrails through continued training on new attack data \cite{vijil-mbert, ayub2024embeddingbasedclassifiersdetectprompt}. Our model abstracts these update procedures into effective repair and retention assumptions. It asks what follows when discovered failures are repaired without losing earlier protection. Retention requires separate justification: \citet{qi2025safety} show how shallow safety alignment can be disrupted by subsequent fine-tuning and study ways to make it more persistent. At the system level, \citet{majumdar2025red} argue for red teaming throughout the AI lifecycle and beyond isolated model failures. Our framework gives a mathematical account of one part of that process, namely the accumulation of protection from discovery and repair.

\paragraph{Scope of coverage guarantees.}

Specifying the domain of permissible attacks drives what a protection guarantee means for the defender. For example, many-shot jailbreaking~\citep{anil2024many} and multi-turn attacks \cite{russinovich2025crescendo,mehrotra2024tap} demonstrate how longer contexts enable new attacks.  On the other hand, formal robustness guarantees can restrict the allowable changes to an input. \citet{kumar2024certifying} introduce erase-and-check, which certifies protection against specified token-level attacks up to a bounded size, subject to the underlying
safety filter's detection of the harmful prompt. This motivates our treatment of an expanding attack surface, where eventual protection against each fixed attack need
not imply complete coverage at any finite time.

In recent work, \citet{vassilev2026robust} shows through incompleteness arguments that no fixed set of guardrails can protect an AI system against all possible attacks. In contrast, we permit the update of defensive guardrails to an AI system. This is a reasonable assumption to make in real-world security practices, where defenses are regularly updated to patch an ongoing feed of newly-found vulnerabilities.

\paragraph{Security investment and recovery.}

The economic value of protection has substantial literature. \citet{gordon2002economics} formulate security investment as a tradeoff between expenditure and expected breach loss.
\citet{ebel2024economics} incorporate a strategic attacker in a defender-led Stackelberg model, characterize optimal investment decisions, and analyze a price of deterrence. Their work already establishes the importance of distinguishing the investment needed to discourage attack from the investment a defender finds optimal. We make discovery \textit{before} compromise and repair \textit{after} compromise separate investment choices, and model attacker response explicitly. This leads to a convex joint investment problem, from which we obtain resource allocation regimes and least-cost deterrent rates.


\paragraph{Stackelberg models in security.}

Our leader--follower structure follows the security-game literature, where a defender commits to a strategy and an attacker responds after observing it~\citep{korzhyk2011stackelberg}. A close application to AI security is the work of
\citet{han2026dynamic}. They model LLM jailbreak interactions as an extensive-form game, combine it with exploration of prompt space, and use Stackelberg reasoning to guide an agentic defense. While their decisions concern prompts and defensive responses \textit{during} LLM interactions, our decisions concern the discovery and repair capabilities purchased \textit{before} the attacker chooses its search effort.


\section{A Model of Testing and Repair}
\label{sec:setting}

Consider an AI system $S$ that is attacked by an adversary and defended by a defender. Let $\X$ be the set of all attack instances. An instance may be a prompt, a conversation, or a prompt together with the environment in which it is used. For example, the same instruction can have different consequences depending on whether the model can invoke a tool. We consider a turn-based setup, where at the beginning of turn $t$, $X_t\subseteq\X$ is the
\emph{active attack surface}: the attacks currently admitted by the defended system $S \oplus D_t$ with $D_t\subseteq\X$ containing the attacks against which
the current defense is effective. We call these attacks \emph{covered}.
The remaining attacks form the residual vulnerable set
\begin{equation}
 V_t=X_t\setminus D_t.
\end{equation}
Thus an attack belongs to $V_t$ when it is both active and still
unresolved. Initially we keep $X_t$ fixed. Later, we allow it to
grow as longer interactions become available.

The definition of effective protection is fixed throughout a given
model. For a deterministic system, it might mean that an attack no
longer produces the specified harmful outcome. For a randomized
system, it might mean that the probability of that outcome is below
a chosen threshold. The set $D_t$ contains the attacks for which this property holds.

\subsection{Attack discovery}
In each round, the attacker sends a few attack instances to the defended system. In that process, the defender discovers a few undefended instances. Call them  $Z_t\subseteq X_t$, where $\varnothing$ represents no new discovered attacks. After processing the report, the defense updates from $D_t$ to $D_{t+1}$. The update may repair the reported attack alone or may protect against additional attacks as well.

We impose two requirements on how the defender incorporates these updates.

\begin{assumption}
\label{ass:update}
For every round, almost surely,
\begin{enumerate}[leftmargin=*,label=(\roman*),nolistsep]
    \item Effective remediation: A reported vulnerability is repaired by the next round, i.e. $Z_t \cap V_t \subseteq D_{t+1}$.
    \item Retention: Subsequent rounds preserve protection already achieved: $ D_t \subseteq D_{t+1}$.
\end{enumerate}
\end{assumption}
A detector may correctly identify a failure without there being an effective patch, and a stored regression test may remain in a database even after the defense starts failing it again. The assumption requires actual repair and actual preservation of that repair.

%
Let $\F_t$ denote all information available right before round $t$. 
Formally, we work on a probability space $(\Omega,\mathcal F,\Prob)$ with the sequence $(\F_t)_{t\in\N}$ being a filtration, that is, an increasing
sequence of information sets. For each attack $x$, membership in $D_t$ is $\F_t$-measurable;
$Z_t$ and $D_{t+1}$ are measurable at the next round. Unless stated
otherwise, the surfaces $X_t$ are deterministic, finite, and
nondecreasing. 

\subsection{Evolving defenses}
\label{sec:static}

For a defense $C$, write $D(C)$ for its covered set. A static incompleteness statement for a class $\C$ has the form
\begin{equation}
 \forall C\in\C,\qquad \X\setminus D(C)\ne\varnothing,
 \label{eq:static}
\end{equation}
that is, every static defense misses at least one attack instance. A version of this result was introduced by \citet{vassilev2026robust}, who proved that no static certifier can cover all possible input prompts violating a particular policy.

We overcome this limitation of static defenses by using feedback from previous rounds to progressively strengthen defenses. This is a more realistic setting for real life guardrails. Suppose $\X=\{x_1,x_2,\ldots\}$ and allowed defenses cover
finite subsets. At round $t$, take
$D_t=\{x_1,\ldots,x_t\}$. Every defense in this sequence is
incomplete, but each $x_i$ is protected from round $i$ onward.
The uncovered attack can change as the defense improves. 

For a finite attack surface, the conclusion is different. Once every
attack has a finite repair time, there is a last repair time.
\begin{proposition}
\label{prop:finitecore}
Let $X \subseteq \mathcal X$ be finite and $(D_t)$ be monotone. If every $x\in X$
belongs to some $D_{t_x}$, then there is a finite $T$ such that
$X\subseteq D_T$. Thus this property cannot hold for a sequence in
a checker class all of whose members leave an element of $X$ uncovered.
\end{proposition}
\begin{proof}
If $X$ is empty, choose $T=0$. Otherwise choose one finite witness
$t_x$ for each $x\in X$ and set $T=\max_{x\in X}t_x$.
Monotonicity gives $x\in D_T$ for every $x\in X$. Applying the
claimed incompleteness property to the defense at time $T$ yields a
contradiction.
\end{proof}
The proposition explains how to interpret our first coverage theorem.
If updates achieve complete coverage on a finite set $X$, those
updates cannot all belong to a class whose every member is incomplete
on that same $X$. There is no contradiction when the finite set is
only a restricted part of a larger universe, or when the update
process is allowed a different class of defenses. Keeping the domain
and checker class explicit resolves the apparent tension.


\section{Complete Coverage of a Finite Attack Surface}
\label{sec:finite}
We now return to the testing process and first keep its attack surface
fixed ($X_t = X \forall t$) and finite ($N := |X| < \infty$). This could represent a bounded family of interactions against
which a defense is to be evaluated. The aim is to understand when
the residual set $V_t$ eventually becomes empty and how long this
may take. To this end, we assume that all attack instances remain discoverable.

\begin{assumption}
\label{ass:discoverable} 
There exists $\epsilon\in(0,1]$ such that
\begin{equation}
 \Prob( x \in Z_t \mid \F_t)\geq\epsilon
 \quad\text{on }\{x\in V_t\},\qquad x\in X.
 \label{eq:finitefair}
\end{equation}
\end{assumption}

Thus, as long as $x$ remains vulnerable, the testing process has a non-zero probability of finding it in the next round. The process can prefer some attacks over others and adapt those preferences over time, provided this lower bound remains available for each unresolved attack.

\begin{theorem}
\label{thm:finite}
Under Assumptions~\ref{ass:update}-\ref{ass:discoverable}, let
$T=\inf\{t\in\N:V_t=\varnothing\}$, with
$\inf\varnothing=\infty$. Then $\Prob(T<\infty)=1$, and for
every integer $t\geq0$,
\begin{equation}
 \Prob(T>t)\leq\min\!\left\{ 1, N(1-\epsilon)^t \right\}.
 \label{eq:finitetail}
\end{equation}
Consequently,
\begin{align}
\E[T]
&\leq \sum_{t=0}^\infty \min\!\left\{1, N(1-\epsilon)^t\right\}
\leq \frac{N}{\epsilon}.
\label{eq:finiteexpectation}
\end{align}
\end{theorem}

\begin{proof}
(See supplementary material for a Lean-verified version) 
We first follow one attack $x\in X$. Write $B_t(x)=\{x\in V_t\}$ for the event that it is still vulnerable at the beginning of round $t$. If $B_t(x)$ occurs, by Assumption~\ref{ass:discoverable} there is at least $\epsilon$ chance of discovering $x$, and effective remediation (Assumption~\ref{ass:update}(i)) then removes it before the next round.
If $B_t(x)$ does not occur, retention ensures that the attack remains
covered. Combining these two cases gives
\begin{equation}
 \Prob(B_{t+1}(x)\mid\F_t)
 \leq\one_{B_t(x)}(1-\epsilon).
 \label{eq:singleattack}
\end{equation}
Here $\one_{B_t(x)}$ is one if $x$ is vulnerable and zero otherwise.
Other updates may also happen to repair $x$, which can only improve
this bound.

Averaging~\eqref{eq:singleattack} over possible histories yields $\Prob(B_{t+1}(x))\leq(1-\epsilon)\Prob(B_t(x))$. Applying the same inequality repeatedly, starting with
$\Prob(B_0(x))\leq1$, gives $\Prob(x\in V_t)\leq(1-\epsilon)^t$. This calculation uses conditional probabilities, and therefore does not require successive rounds to be independent.

We next consider the whole surface. Because $X$ is fixed and covered attacks stay covered, $T>t$ holds exactly when $V_t$ is nonempty. The probability that at least one attack remains is no greater than the sum of the individual survival probabilities:
$$ \Prob(T>t) = \Prob(V_t\ne\varnothing)
\leq \sum_{x\in X}\Prob(x\in V_t) 
\leq N(1-\epsilon)^t, $$
giving us \eqref{eq:finitetail}. As $t$ increases, the bound tends to zero, so the probability of
never completing the repairs is zero.

Finally, the number of rounds before completion can be written as $T=\sum_{t\geq0}\one_{\{T>t\}}$, so that $\E[T]=\sum_{t\geq0}\Prob(T>t)$. Substituting the tail bound and summing the geometric series proves~\eqref{eq:finiteexpectation}.
\end{proof}


A useful probability guarantee comes directly from the tail bound in Eq.~\eqref{eq:finitetail}. Since
$1-u\leq e^{-u}$, we have $\Prob(T>t)\leq\min\{1,N \exp (-\epsilon t) \}$. Thus, for any
$\delta\in(0,1)$, choosing
\begin{equation}
 t\geq\left\lceil\frac{\log(N/\delta)}{\epsilon}\right\rceil
 \label{eq:confidence}
\end{equation}
ensures at least probability $1-\delta$ of covering the entire attack surface by round $t$.

\begin{remark}
There are two simple boundary considerations. First, if $D_0$
is deterministic, only initially unresolved attacks need to be
counted, so $N$ may be replaced by $|X\setminus D_0|$. Second, if $X=\varnothing$, then $T=0$; more generally,
completion is immediate whenever $X\subseteq D_0$.
Note also that the theorem separates eventual coverage from the speed of improvement. Even a small $\epsilon>0$ is enough for eventual repair, but it may require a large number of rounds. Thus, effective exploration is important as it shortens this waiting time.
\end{remark}

\section{Coverage under Growing Contexts}
\label{sec:long}
State of the art foundation models allow for very long context windows. To model this, we now allow context to expand infinitely in our setup. Increasing the
maximum length of an interaction admits attacks that could not have been expressed within the earlier limit. In this section, we first consider a generic extension of Theorem~\ref{thm:finite} under a growing context window, then generalize it further under realistic defense and discovery mechanisms. 

\subsection{Complete coverage}
Let $X^{(1)}\subseteq X^{(2)}\subseteq\cdots$ be finite attack
surfaces, with the index $\ell$ representing a context-length limit. Write
$X^{(\infty)}=\bigcup_{\ell\geq1}X^{(\ell)}\subseteq\X$. Since each set is finite, their union is countable. Take a nondecreasing sequence of indices $ \{L_t: 1 \leq t < \infty \}$, and set $X_t=X^{(L_t)}$. The \emph{admission time}
\begin{equation}
 a(x)=\min\{t:x\in X_t\}
\end{equation}
is the first round in which attack $x$ is available. It is finite for every $x\in X^{(\infty)}$. Note that this construction is literal for attack instances being bounded-length text over a finite vocabulary with $a(x) = |x|$. For conversations or tool-using interactions, an attack instance includes more than the latest user prompt. Additional fields may describe earlier dialogue turns, retrieved documents, tool outputs, or the tools and permissions available to the model. For example, a short instruction may become harmful only when paired with a malicious passage in a retrieved document or an earlier tool response. 

The appropriate conclusion now concerns each attack separately. Once an attack is admitted, we enquire if the defense can eventually protect against it.

\begin{theorem}
\label{thm:long}
Suppose Assumption~\ref{ass:update} holds and each $x\in X^{(\infty)}$ is discoverable, that is, there exists $\epsilon_x\in(0,1]$ such that
\begin{equation}
 \Prob( x \in Z_t \mid\F_t)\geq\epsilon_x
 \quad\text{on }\{x\in V_t\},\qquad t\geq a(x).
\end{equation}
Then, for $t\geq a(x)$,
\begin{equation}
 \Prob(x\notin D_t)\leq(1-\epsilon_x)^{t-a(x)}.
 \label{eq:longtail}
\end{equation}
Moreover,
\begin{equation}
 \Prob\!\left(\begin{gathered}
 \forall x\in X^{(\infty)}\ \exists T_x<\infty:\\
 \forall t\geq T_x,\ x\in D_t
 \end{gathered}\right)=1.
 \label{eq:pointwise}
\end{equation}
Equivalently, $X^{(\infty)}\subseteq\bigcup_tD_t$ almost surely, with equality if the defended sets are themselves restricted to $X^{(\infty)}$.
\end{theorem}

\begin{proof}
(See supplementary material for a Lean-verified version) 
Fix an attack $x$. Before its admission, the testing process is not
required to find it. Starting at round $a(x)$, it remains
active and the proof argument from Theorem~\ref{thm:finite} applies with discovery bound $\epsilon_x$. After $t-a(x)$ opportunities for repair, it gives
\[
 \Prob(x\notin D_t)\leq(1-\epsilon_x)^{t-a(x)}.
\]
The events $\{x\notin D_t\}$ decrease with time because repairs
are retained. Their probabilities tend to zero, so their intersection,
the event that $x$ is never covered, has probability zero.

It remains to show that this conclusion holds for all attacks on
the same realization of the feedback process. Let $E_x$ be the
event that $x$ is eventually covered. We have just shown that
$\Prob(E_x^c)=0$ for each $x$. Since the universe is countable,
\[
 \Prob\!\left(\bigcup_{x\in X^{(\infty)}}E_x^c\right)
 \leq\sum_{x\in X^{(\infty)}}\Prob(E_x^c)=0.
\]
Retention (Assumption~\ref{ass:update}(ii)) makes each attack's coverage permanent once it occurs. Finally,
membership in $\bigcup_tD_t$ means membership at some finite time,
which under retention is exactly eventual permanent coverage. This
proves the equivalent set formulation as well.
\end{proof}

Each fixed context level also has a finite completion time. The
reason is that it contains only finitely many attacks, even though
the complete universe may be infinite.
\begin{corollary}
\label{cor:local}
Almost surely, for every $\ell$ there is a finite $T_\ell$ such
that $X^{(\ell)}\subseteq D_t$ for all $t\geq T_\ell$.
\end{corollary}
\begin{proof}
On the single probability-one event in~\eqref{eq:pointwise}, take
the maximum of the finitely many witness times for each level; use
zero for an empty level. This gives the assertion simultaneously for
all levels.
\end{proof}

\subsection{Robustness toward similar attacks}
\label{sec:regions}
So far, discovering an attack has guaranteed protection against that attack alone. This is a conservative description of what a useful defensive patch $D_{t+1} \setminus D_t$ might accomplish. Several prompts can exploit the same failure mechanism, and changing that \textit{mechanism} can protect against all of them. The next result describes how this kind of generalization can affect coverage.

For a finite attack surface $X$, let $R_1,\ldots,R_m \subseteq \X$ so that $\bigcup_{i=1}^m R_i = X$. We call them \emph{patch regions}. A region represents a family of attacks that can be addressed by the same repair. The regions may overlap because an attack can belong to more than one such family. Define $H_{i,t}=\{R_i\cap V_t\ne\varnothing\}$, the event that region $i$ still contains an unresolved attack. We place the following assumptions on these quantities.

\begin{assumption}
\label{ass:region}
For the patch regions, the following holds.
\begin{enumerate}[label=(\roman*),leftmargin=*,nolistsep]
    \item Discoverability: for an unresolved region, the discovery process has non-zero probability $\eta_i\in(0,1]$ of finding a vulnerable member:
    \begin{equation}
     \Prob(Z_t\cap R_i\cap V_t \neq \varnothing \mid \F_t)\geq\eta_i
     \quad\text{on }H_{i,t}.
     \label{eq:regiondiscover}
    \end{equation}
     \item Effective remediation: discovery leads to a repair covering the whole region:
    \begin{equation}
     Z_t\cap R_i\cap V_t \neq \varnothing\ \Longrightarrow\
     R_i\cap V_{t+1}=\varnothing\quad\text{almost surely}.
     \label{eq:regionrepair}
    \end{equation}
\end{enumerate}
\end{assumption}

Part (i) and (ii) are modified versions of Assumptions~\ref{ass:discoverable} and \ref{ass:update}(i), respectively. These conditions allow us to establish a generalization of Theorem~\ref{thm:finite}.

\begin{theorem}
\label{thm:regions}
Under Assumption~\ref{ass:update} and \ref{ass:region}, $T=\inf\{t:V_t=\varnothing\} < \infty$ almost surely, with
\begin{align}
 \Prob(T>t)&\leq\min\!\left\{1,\sum_{i=1}^m(1-\eta_i)^t\right\},
 \label{eq:regiontail}\\
 \E[T]&\leq\sum_{i=1}^m\eta_i^{-1}.
\end{align}
\end{theorem}

\begin{proof}
(See supplementary material for a Lean-verified version) 
Fix a region $R_i$. If it has already been cleared, retention keeps
it clear. If it still contains a vulnerability, there is probability
at least $\eta_i$ of discovering a vulnerable member and thereby
clearing the whole region. Therefore
\[
 \Prob(H_{i,t+1}\mid\F_t)\leq\one_{H_{i,t}}(1-\eta_i).
\]
Taking expectations and repeating this inequality across rounds gives
$\Prob(H_{i,t})\leq(1-\eta_i)^t$.

Since the regions cover $X$, any remaining vulnerable attack must
belong to an unresolved region. Consequently,
$\{T>t\}=\bigcup_{i=1}^m H_{i,t}$. Adding the region survival
probabilities and bounding the result by one proves
\eqref{eq:regiontail}. The bound tends to zero, establishing
almost-sure completion. Summing it over rounds, as in the finite
theorem, gives $\E[T]\leq\sum_i1/\eta_i$. Regions that are
empty or initially cleared may be left out of these sums.
\end{proof}

Theorem~\ref{thm:finite} is a special case where regions are singleton sets. With larger regions, a patch can complete several repairs after one discovery, with the gain depending on the number of regions and their discoverability. A large region that is rarely encountered may still take a long time to repair.

As an example, suppose $N=mr$ initially vulnerable attacks are partitioned into $m \geq 1$ regions of size $r > 1$ each. Each round independently samples one of the $N$ attacks uniformly. If a repair covers only the sampled attack, the expected completion time is $NH_N, H_n:=\sum_{k=1}^n1/k$. When $k$ attacks are unresolved, a sample finds one of them with probability $k/N$, so the expected wait for the next new repair is $N/k$. Adding these waiting times as $k$ decreases from $N$ to one gives $NH_N$. For region-level repairs, the same argument applies to each of the $m$ regions, giving $mH_m$. 


\begin{remark}
It is possible to weaken the assumption of perfect repair over a region (Assumption~\ref{ass:region}(ii)). Suppose discovery in an unresolved region occurs with probability at least $\eta_i$, and, conditional on this discovery and the preceding history, the patch clears the region with probability at least $r_i>0$. The chance of
both discovering and clearing the region is then at least $\eta_i r_i$. Provided unsuccessful patches preserve earlier protection, the same proof works with $\eta_i r_i$ in place of $\eta_i$. Repeated opportunities can therefore compensate for imperfect repair, although they increase the waiting time. The conditional formulation is what justifies the product: multiplying unrelated marginal discovery and repair rates would not suffice.
\end{remark}

\subsection{Several discovery mechanisms}
\label{sec:mechanisms}
An organization can receive findings from automated testing, human red teams, and deployment incidents. These discovery mechanisms have their own strengths and blind spots. 
We now ask when these contributions, taken together, supply enough exploration for the coverage results.

\subsubsection{Accumulating opportunities for repair}
To begin with, we relax the earlier requirement of a fixed positive discovery
bound in every round. A mechanism may receive less attention over time, or an attack may become progressively harder to discover. The useful quantity is then the accumulation of discovery opportunities across rounds. To study it, we first state the time-varying version of the probability argument used in the preceding proofs.

Let $B_t$ be the event that an attack region remains unresolved at round
$t$, and the conditional chance of resolving it is at least $q_t$. Thus, its conditional chance of surviving that round is at most $1-q_t$, multiplying which across rounds provides the following bound. 
%
\begin{lemma}
\label{lem:survival}
Let $B_t\in\F_t$ for integers $t\geq a$, with $B_{t+1}\subseteq B_t$ up to null sets. Suppose $q_t\in[0,1]$ satisfy

\begin{equation}
 \Prob(B_{t+1}\mid\F_t)\leq\one_{B_t}(1-q_t).
 \label{eq:conditional}
\end{equation}
Then, for $n\geq a$,
\begin{align}
\Prob(B_n)
\leq\Prob(B_a) \prod_{t=a}^{n-1}(1-q_t)
\leq\Prob(B_a) \exp \! \left( -\sum_{t=a}^{n-1}q_t \right).
\label{eq:exponential}
\end{align}
If $\sum_{t=a}^{\infty}q_t=\infty$, then
$\Prob(\bigcap_{n\geq a}B_n)=0$.
\end{lemma}
\begin{proof}
Average~\eqref{eq:conditional} over the information available before
round $t$. The tower property of conditional expectation gives
\[
 \Prob(B_{t+1})\leq(1-q_t)\Prob(B_t).
\]
Repeating the inequality from $a$ through $n-1$, then applying $1-u\leq e^{-u}$ to each
factor obtains \eqref{eq:exponential}.

If $\sum_t q_t$ diverges, the exponential bound tends to zero. The events $B_n$ decrease, their intersection describing attacks that always remain unresolved. Continuity of probability for decreasing events identifies its probability with $\lim_n\Prob(B_n)=0$.
\end{proof}

Thus, discovery probability may decrease without preventing eventual repair. A sufficient condition for an unresolved attack region to be repaired eventually with probability 1 is that its cumulative conditional repair-probability lower bound diverges: \(\sum_t q_t=\infty\).

\subsubsection{Combining the mechanisms}
Let $\mathcal K$ be a finite set of discovery mechanisms. In each round, one mechanism $K_t$ is selected with conditional probabilities
\begin{equation}
 \alpha_{k,t}=\Prob(K_t=k\mid\F_t),\qquad
 \sum_{k\in\mathcal K}\alpha_{k,t}=1.
\end{equation}
The probabilities describe how testing effort is allocated and may
depend on the history. For each selected mechanism, let $p_{k,t}(x)$
be an $\F_t$-measurable lower bound on its chance of reporting an
active vulnerable attack $x$, conditional on both $\F_t$ and
$K_t=k$. We only need this bound when $\alpha_{k,t}>0$.
Weighting the mechanism-specific bounds by their selection probabilities
gives
\begin{equation}
 \Prob(x \in Z_t \mid \F_t)\geq
 \sum_{k\in\mathcal K}\alpha_{k,t}p_{k,t}(x)
 \quad\text{on }\{x\in V_t\}.
 \label{eq:mixture}
\end{equation}
Thus a mechanism contributes through both its ability to discover $x$
and the frequency with which it is used. 
We can now apply Lemma~\ref{lem:survival} to the combined process.
\begin{theorem}
\label{thm:mechanisms}
Suppose Assumption~\ref{ass:update} holds, and the discovery probability for a given attack $x$ satisfies ~\eqref{eq:mixture}. Then for $n\geq a(x)$,
\begin{equation}
 \Prob(x\notin D_n)\leq\prod_{t=a(x)}^{n-1}(1-q_t(x)).
 \label{eq:mechanismtail}
\end{equation}
In particular, $\sum_{t=a(x)}^\infty q_t(x)=\infty$ implies eventual permanent coverage of $x$ almost surely.
\end{theorem}

\begin{proof}
(See supplementary material for a Lean-verified version) 
Follow $x$ from its admission time and put $B_t=\{x\notin D_t\}$.
Equation~\eqref{eq:mixture} and the lower bound $q_t(x)$ give a
chance of at least $q_t(x)$ of discovering it whenever $B_t$ occurs.
Effective remediation and retention therefore yield the survival
inequality in Lemma~\ref{lem:survival}. Its product bound gives
\eqref{eq:mechanismtail}, and divergence of the sum implies that $x$
is eventually covered with probability 1.

For a countable surface, intersect the probability-one events of
eventual coverage, one for each attack, as in the proof of
Theorem~\ref{thm:long}. For a fixed finite surface, there are only
finitely many repair times on this event. Taking their maximum, as
in Proposition~\ref{prop:finitecore}, gives finite-time completion
of the entire surface.
\end{proof}


\begin{remark}
Merely including a discovery mechanism in an organization's defensive toolkit does not ensure persistent exploration. A mechanism may be capable of finding an attack yet be used so rarely that its cumulative discovery contribution is finite. Also, with finitely many mechanisms, an infinite combined cumulative contribution for a fixed attack requires at least one mechanism to contribute an infinite amount for that attack. Complementarity allows that mechanism to differ from one attack to another.

The weighted-sum formula describes selection of one mechanism per round. If several mechanisms operate simultaneously, one instead needs the probability that at least one reports the attack. Dependence between their reports must then be accounted for rather than adding their marginal probabilities. An incident mechanism has a further operational interpretation: its report follows an actual failure. It can contribute to preventing recurrence, while leaving the initial harm outside the eventual-coverage guarantee. 
\end{remark}

\section{A Stackelberg Model for Strategic Defense}
\label{sec:stackelberg}

The preceding results identify conditions under which continued testing and
repair ensure protection for the defender. They leave open an economic question:
how much should the defender invest in these capabilities? Investment also changes the attacker's incentives. An exploit that is technically possible may not be worth pursuing if the defender is likely to discover it first or repair it shortly after compromise.

We model this interaction as a defender-led Stackelberg game. The defender first commits to its testing and response capabilities. The attacker observes this commitment and chooses its search effort. The defender anticipates that response when deciding how much
protection to purchase.

\subsection{The discovery race}

Let \(a\geq0\) denote the attacker’s committed search effort and \(r>0\) the discovery rate per unit of effort. Thus, effort \(a>0\) produces an exponential discovery time with rate \(ra\). Let $v>0$ denote the attacker's benefit per unit of time that the exploit remains usable, and $c_A>0$ the cost per unit of committed search effort. Let $L>0$ be the defender's loss per unit of compromise time. 

The defender chooses additional resource amounts $ d=(g,h,i) \in [0,\infty)^3$, which are allocated to automated testing, human red teaming, and incident response, respectively. Given their respective costs per unit $K_G, K_H, K_I >0$, the total investment cost is
\begin{equation}
C_D(d)=k_Gg+k_Hh+k_Ii.
\label{eq:strategic-defense-cost}
\end{equation}
These investments determine the proactive discovery rate $m(d)$
and the repair rate after compromise $\mu(d)$:
\begin{align}
m(d)&=m_0+q_Gg+q_Hh,
\label{eq:strategic-discovery}\\
\mu(d)&=\delta+q_Ii.
\label{eq:strategic-repair}
\end{align}
The coefficients $q_G,q_H,q_I$ are nonnegative and measure the
effectiveness of investment in the respective discovery mechanisms. The positive baselines $m_0$ and
$\delta$ represent discovery and repair capabilities already in
place. Thus, zero additional investment does not mean an absence
of defense.

Automated and human testing contribute to the same proactive
discovery rate. Incident response has a different role: it shortens
the period during which a successful exploit remains usable.
The model therefore distinguishes preventing compromise from
limiting its duration.


After observing $d$, the attacker commits effort $a\geq0$. Its time to discovering and exploiting the vulnerability is Exponential$(ra)$. Independently, the defender discovers and neutralizes the vulnerability at rate $m(d)$. If the defender acts first, the opportunity ends without compromise. If the attacker acts first, the system enters a compromised state. Thus, the time from compromise to effective repair is distributed as Exponential$(\mu(d))$. In this model, $m(d)$ governs the period before compromise and $\mu(d)$ is the total repair rate afterward.

The probability of compromise is
\begin{equation}
p(a,d)=\frac{ra}{ra+m(d)}.
\label{eq:strategic-win}
\end{equation}
Conditional on compromise, the expected time until repair is $1/\mu(d)$, so the expected compromise time per opportunity is
$$ W(a,d)=\frac{1}{\mu(d)} \frac{ra}{ra+m(d)}.$$
%

The attacker pays $c_Aa$ upfront to commit its search effort. So its expected payoff is
$$ U_A(a;d)= \frac{v}{\mu(d)}\frac{ra}{ra+m(d)}-c_Aa. $$
The timing of this cost is part of the model: it is paid for the
chosen effort, rather than accumulated over the random search time.
Choosing $a=0$ gives payoff zero.

\subsection{The attacker's best response}

The attacker weighs the expected benefit of finding the vulnerability
first against the cost of search. The following result makes this
tradeoff explicit.

\begin{theorem}
\label{thm:strategic-attacker}
For every defensive investment $d$, the attacker has a unique
best response
\begin{equation}
a^*(d)=\frac{1}{r}
\left[
\sqrt{\frac{vrm(d)}{c_A\mu(d)}}-m(d)
\right]_+,
\label{eq:strategic-best-response}
\end{equation}
where $[u]_+=\max\{u,0\}$. The attacker is deterred when
$$ a^*(d)=0
\quad\Longleftrightarrow\quad
m(d)\mu(d)\geq\theta, $$
with $\theta := vr/c_A$. The compromise probability and expected compromise time at the best response are
\begin{align*}
p^*(d)
&=\left[1-\sqrt{\frac{m(d)\mu(d)}{\theta}}\right]_+,\\
W^*(d)
&=\frac{1}{\mu(d)}
\left[1-\sqrt{\frac{m(d)\mu(d)}{\theta}}\right]_+.
\end{align*}
\end{theorem}

\begin{proof}
Fix $d$ and write $m=m(d)$ and $\mu=\mu(d)$. Differentiating
the attacker's payoff gives
\[
\frac{\partial U_A}{\partial a}
=\frac{vrm}{\mu(ra+m)^2}-c_A,
\qquad
\frac{\partial^2U_A}{\partial a^2}
=-\frac{2vr^2m}{\mu(ra+m)^3}<0.
\]
The payoff is strictly concave and tends to $-\infty$ as
$a\to\infty$. Its unique maximizer is zero if the derivative at
zero is nonpositive. Otherwise, it is the positive solution of
the first-order condition. Solving gives
\eqref{eq:strategic-best-response}.

The derivative at zero is nonpositive exactly when
$vr/(\mu m)\leq c_A$, which gives the deterrence condition.
When $a^*>0$, the first-order condition implies
\[
ra^*+m=\sqrt{\frac{\theta m}{\mu}}.
\]
Substitution into \eqref{eq:strategic-win} gives
$p^*=1-\sqrt{m\mu/\theta}$. When $a^*=0$, compromise probability
and expected compromise time are both zero.
\end{proof}

To interpret the deterrence condition $m(d) \mu(d) \geq \theta$, once the product of the discovery and repair rates reaches $\theta$, the attacker’s optimal search effort is zero. Greater attacker value $v$ or search productivity $r$ makes deterrence harder; greater attack cost $c_A$ makes it easier. The defender can reach the threshold by improving proactive discovery, accelerating repair, or combining the two. Faster repair reduces the benefit of winning the discovery race and, therefore, reduces the proactive capability required for deterrence.

We note that attacker effort and defensive outcomes may not move in the same direction. As proactive discovery improves, an active attacker may initially increase effort to remain competitive. Nevertheless, the formulas for $p^*$ and $W^*$ show that compromise probability and expected compromise time are nonincreasing in both defensive rates.

\subsection{The defender's investment problem}

The defender anticipates the attacker's best response and solves
\begin{equation}
d^* \in \operatorname*{arg\,min}_{d\geq0}J(d),\quad
J(d) = C_D(d)+LW^*(d).
\label{eq:strategic-leader}
\end{equation}
Investment cost and expected loss are measured on the same
opportunity basis. In particular, $C_D$ represents the cost
allocated to providing the chosen capabilities for the modeled
opportunity. 

\begin{proposition}
\label{prop:strategic-leader}
The objective in \eqref{eq:strategic-leader} is continuous and
convex on $[0,\infty)^3$ and attains a minimum. Hence a
Stackelberg equilibrium exists. The optimal investment allocation
need not be unique.
\end{proposition}

\begin{proof}
Write the expected loss as
\[
F(m,\mu)=\frac{L}{\mu}
\left[1-\sqrt{\frac{m\mu}{\theta}}\right]_+,
\qquad m,\mu>0.
\]
In the region where the attacker is active, let
$s=\sqrt{m\mu/\theta}<1$. Differentiation gives
\[
F_{mm}=\frac{Ls}{4m^2\mu}>0,
\qquad
F_{\mu\mu}=\frac{L}{\mu^3}
\left(2-\frac{3s}{4}\right),
\]
and
\[
\det\nabla^2F
=\frac{L^2s(2-s)}{4m^2\mu^4}>0.
\]
Thus the Hessian is positive definite in this region.
In the interior of the deterrence region, $F$ is zero.
Along any line crossing the boundary $m\mu=\theta$, the
positive-part operation makes the derivative jump upward;
at tangency, there is no downward jump. The restriction to
every line segment in the positive quadrant is therefore convex.

The rates $m(d)$ and $\mu(d)$ are affine, so composing $F$
with these rates and adding the linear investment cost preserves
convexity. Positive baseline rates ensure continuity. Finally,
\[
J(d)\geq k_Gg+k_Hh+k_Ii.
\]
Because all three cost coefficients are positive, the objective
tends to infinity as investment becomes unbounded. It therefore
attains a minimum. Combining a minimizer with the attacker's
unique best response gives an equilibrium.
\end{proof}

The pair $(d^*,a^*(d^*))$ is a Stackelberg equilibrium, and Proposition~\ref{prop:strategic-leader} proves its existence. Note that deterrence need not be the defender's preferred outcome. Even when sufficient investment could eliminate the attacker's incentive, the defender may find the required capability more expensive than the expected loss it would prevent.

\subsection{Choosing between d mechanisms}


To perform proactive defensive testing, automated testing and human red teaming are substitutes. Given that at least one of $q_G,q_H$ is positive, the least cost of raising the proactive discovery rate by one unit is
\begin{equation}
\kappa=
\min_{\ell\in\{G,H\}:\,q_\ell>0}
\frac{k_\ell}{q_\ell}.
\label{eq:strategic-effective-cost}
\end{equation}
For a desired proactive discovery rate \(m\geq m_0\), the defender must purchase an increase of \(m-m_0\) above its baseline rate. Allocating this increase to the mechanism with the lowest cost per unit of discovery gives a minimum proactive expenditure of \(\kappa(m-m_0)\).



In contrast, proactive testing and incident response affect different parts of the attacker's payoff. Their allocation can be characterized by the marginal reduction in expected loss. When the attacker is active, define
\[
s=\sqrt{\frac{m\mu}{\theta}},
\qquad
b=\frac{Ls}{2m\mu},
\qquad
e=\frac{L}{\mu^2}\left(1-\frac{s}{2}\right).
\]
Here $b=-F_m$ is the benefit of increasing proactive discovery,
and $e=-F_\mu$ is the benefit of increasing repair speed.
Both account for the attacker's adjustment in effort.

When $m\mu>\theta$, set $b=e=0$. At the deterrence boundary,
the appropriate marginal values have the form
\[
b=\tau\frac{L}{2m\mu},
\qquad
e=\tau\frac{L}{2\mu^2},
\qquad \tau\in[0,1].
\]
The same $\tau$ appears in both expressions because the two
marginal values come from a single subgradient of $F$.

By convexity, an investment $d=(g,h,i)$ is optimal if and only
if these marginal values can be chosen to satisfy
\begin{align}
q_Gb&\leq k_G,
&g(k_G-q_Gb)&=0,\nonumber\\
q_Hb&\leq k_H,
&h(k_H-q_Hb)&=0,\nonumber\\
q_Ie&\leq k_I,
&i(k_I-q_Ie)&=0.
\label{eq:strategic-mechanism-conditions}
\end{align}
For each funded mechanism, marginal benefit equals marginal cost. For an unfunded mechanism, marginal benefit is no greater than marginal cost. Depending on these comparisons, the defender may invest in proactive testing, incident response, both, or neither.

\subsection{Defender's investment regimes}

The defender's investment strategy changes based on the values of the proactive investment cost \(\kappa\). To obtain explicit expressions for these cost thresholds, we begin by holding the repair rate \(\mu>0\) fixed.
This is a restriction of the preceding game: the general model allows investment in
repair as well, whereas the following result describes the decision at a given repair capability.

Assume at least one proactive mechanism has positive effectiveness, and let $\kappa$ be its least effective cost from \eqref{eq:strategic-effective-cost}. Define
\[
M=\frac{vr}{c_A\mu},
\qquad
A=\frac{L}{\mu}.
\]
Here $M$ is the proactive discovery rate required for deterrence,
and $A$ is the expected loss conditional on compromise.
After omitting any fixed response cost, the defender's problem is
\begin{equation}
\min_{m\geq m_0}
\left\{
\kappa(m-m_0)+A\left[1-\sqrt{\frac{m}{M}}\right]_+
\right\}.
\label{eq:strategic-scalar}
\end{equation}

\begin{theorem}
\label{thm:strategic-cost-regimes}
If $m_0\geq M$, baseline proactive capability already deters
the attacker and the unique optimal rate is $m^*=m_0$.
If $m_0<M$, the unique optimal proactive rate is
\begin{equation*}
m^*=
\left\{\begin{array}{l@{\qquad}r}
m_0,
&\displaystyle \kappa\geq\frac{A}{2\sqrt{Mm_0}},\\[2mm]
\displaystyle\frac{A^2}{4\kappa^2M},
&\displaystyle \frac{A}{2M}<\kappa<
\frac{A}{2\sqrt{Mm_0}},\\[3mm]
M,
&\displaystyle 0<\kappa\leq\frac{A}{2M}.
\end{array}\right.
\end{equation*}
These cases correspond to no additional proactive investment,
additional investment with continued attack, and full deterrence.
\end{theorem}

\begin{proof}
If $m_0\geq M$, expected compromise loss is already zero and
additional investment only increases cost. Otherwise, no rate
above $M$ can be optimal, since it costs more than $M$ without
reducing loss further.

On $[m_0,M)$, let $\Phi(m)$ denote the objective in
\eqref{eq:strategic-scalar}. Its derivatives are
\[
\Phi'(m)=\kappa-\frac{A}{2\sqrt{Mm}},
\qquad
\Phi''(m)=\frac{A}{4\sqrt{M}\,m^{3/2}}>0.
\]
If $\Phi'(m_0)\geq0$, the minimum occurs at $m_0$.
If $\Phi'(M^-)\leq0$, the objective decreases up to $M$,
where the minimum occurs. In the remaining case, its unique
minimum is the interior solution of $\Phi'(m)=0$, namely
$A^2/(4\kappa^2M)$. These conditions give the stated thresholds,
including the boundary cases.
\end{proof}

When proactive capability is expensive, the defender accepts the residual risk under its baseline investment. At intermediate costs, it purchases additional protection but stops before eliminating the attacker's incentive. When capability is sufficiently inexpensive, it purchases exactly enough to deter search.

For a baseline that does not already deter, the condition for optimal deterrence can also be written as
\[
\kappa\leq\frac{Lc_A}{2vr}.
\]
A greater loss to the defender increases its willingness to
pay for deterrence. A greater cost of attacker search also
favors deterrence, while greater attacker value or search
productivity makes it harder to justify. 





\section{Joint Investment in Discovery and Repair}
\label{sec:joint-investment}

Holding repair capability fixed makes the cost regimes easy to
interpret, but a defender can often invest in both finding
vulnerabilities and repairing successful exploits. We now characterize
this joint decision. In particular, we distinguish the least cost of
deterring an attack from the decision about whether deterrence is
worth purchasing.

Suppose at least one proactive mechanism has positive effectiveness
and $q_I>0$. Recall that $\kappa$ is the least cost of increasing
the proactive discovery rate by one unit, and define
$\eta=k_I/q_I$ as the corresponding cost for the repair rate.
The investment problem reduces to
\begin{equation}
\begin{aligned}
&\min_{m\geq m_0,\,\mu\geq\delta}\ \Phi(m,\mu),\\
\Phi(m,\mu)={}&\kappa(m-m_0)+\eta(\mu-\delta)\\
&+\frac{L}{\mu}\left[1-\sqrt{\frac{m\mu}{\theta}}\right]_+,
\end{aligned}
\label{eq:joint-objective}
\end{equation}
where $\theta=vr/c_A$. The rates $m$ and $\mu$ can be implemented
through the least expensive proactive mechanism and incident response.
If a capability cannot be increased, its rate remains fixed and the
problem reduces to the corresponding boundary case.

\subsection{The least cost of deterrence}

Deterrence requires $m\mu\geq\theta$. Once this condition is
satisfied, additional capability does not further reduce the
modeled compromise loss. The cheapest deterrent allocation therefore
balances the costs of the two ways to reach this threshold.

\begin{proposition}[Least-cost deterrence]
\label{prop:least-cost-deterrence}
If $m_0\delta\geq\theta$, the baseline allocation deters attack
at zero incremental cost. Otherwise, the unique least-cost
deterrent rates are
\begin{align}
m_D&=\min\left\{\frac{\theta}{\delta},
       \max\left\{m_0,\sqrt{\frac{\eta\theta}{\kappa}}\right\}\right\},
\label{eq:least-cost-deterrence}\\
\mu_D&=\frac{\theta}{m_D}.
\nonumber
\end{align}
Their cost is
$C_D^{\mathrm{det}}=\kappa(m_D-m_0)+\eta(\mu_D-\delta)$.
If both rates exceed their baselines, this simplifies to
\begin{equation}
C_D^{\mathrm{det}}=2\sqrt{\kappa\eta\theta}
-\kappa m_0-\eta\delta.
\label{eq:deterrence-cost-interior}
\end{equation}
\end{proposition}

\begin{proof}
Suppose $m_0\delta<\theta$. Positive investment costs imply
that a least-cost deterrent satisfies $m\mu=\theta$:
any strict excess permits a reduction in a rate above its baseline.
Substituting $\mu=\theta/m$ leaves
\[
\min_{m\in[m_0,\theta/\delta]}
\left\{\kappa(m-m_0)+\eta\left(\frac{\theta}{m}-\delta\right)\right\}.
\]
The derivative is $\kappa-\eta\theta/m^2$ and the second
derivative is $2\eta\theta/m^3>0$. Restricting the unconstrained
minimizer $\sqrt{\eta\theta/\kappa}$ to the feasible interval
gives the result.
\end{proof}

When proactive capability is expensive relative to repair, the
solution has $m_D=m_0$ and relies on faster repair. When repair
is expensive, it has $\mu_D=\delta$ and relies on discovery.
Between these cases, both receive investment. This proposition
answers how to deter at least cost; it does not yet establish
that deterrence minimizes investment plus expected loss.

\subsection{When each investment regime is optimal}

For an allocation at which attack continues, write
\begin{equation}
s=\sqrt{\frac{m\mu}{\theta}},\qquad
b=\frac{Ls}{2m\mu},\qquad
e=\frac{L}{\mu^2}\left(1-\frac{s}{2}\right).
\label{eq:joint-marginal-benefits}
\end{equation}
The quantities $b$ and $e$ are the marginal reductions in expected
loss from increasing discovery and repair, respectively. They
include the attacker's response to the changed rates. Comparing
these benefits with $\kappa$ and $\eta$ identifies the equilibrium.

\begin{proposition}
\label{prop:joint-regimes}
If $m_0\delta\geq\theta$, no additional investment is optimal.
If $m_0\delta<\theta$, let $b_0,e_0$ be the values in \eqref{eq:joint-marginal-benefits} at the baseline. Then the following characterize the optimal rates.
\begin{enumerate}[label=(\roman*),leftmargin=*,nolistsep]
\item The baseline is optimal if and only if
      $\kappa\geq b_0$ and $\eta\geq e_0$.
\item Deterrence is optimal if and only if
\begin{equation}
\kappa\leq\frac{L}{2\theta}
\quad\text{or}\quad
\eta\leq\frac{Lm_0^2}{2\theta^2}.
\label{eq:joint-deterrence-region}
\end{equation}
      In this case the rates are $(m_D,\mu_D)$ from
      Proposition~\ref{prop:least-cost-deterrence}.
\item Otherwise, attack continues with additional investment in
      one or both capabilities. A proactive-only solution has
\[
\mu^*=\delta,\qquad
m^*=\frac{L^2}{4\kappa^2\theta\delta}>m_0,
\]
      and is optimal precisely when $m^*\delta<\theta$ and
      $\eta\geq e(m^*,\delta)$.
      A response-only solution has $m^*=m_0$ and the unique
      $\mu^*\in(\delta,\theta/m_0)$ solving
\[
\eta=\frac{L}{(\mu^*)^2}
\left(1-\frac12\sqrt{\frac{m_0\mu^*}{\theta}}\right),
\]
and is optimal when $\kappa\geq b(m_0,\mu^*)$. A solution funding both capabilities while attack continues is
\begin{equation}
\begin{split}
s^*&=\frac{L}{2\kappa\theta},\\
\mu^*&=\sqrt{\frac{L(1-s^*/2)}{\eta}},\qquad
m^*=\frac{\theta(s^*)^2}{\mu^*}.
\end{split}
\label{eq:joint-active-solution}
\end{equation}
      This solution applies when $s^*<1$, $m^*>m_0$, and
      $\mu^*>\delta$.
\end{enumerate}
Equalities with a baseline describe the boundaries between
investment regimes.
\end{proposition}

\begin{proof}
By convexity of \eqref{eq:joint-objective}, the marginal optimality
conditions are necessary and sufficient. In the active region they
are
\[
\begin{aligned}
\kappa&\geq b, & (m-m_0)(\kappa-b)&=0,\\
\eta&\geq e, & (\mu-\delta)(\eta-e)&=0.
\end{aligned}
\]
At the baseline they give the first claim. If only discovery is
increased, solving $\kappa=b$ gives the stated $m^*$; the unused
repair investment requires $\eta\geq e$. If only repair is
increased, solving $\eta=e$ and checking $\kappa\geq b$ gives
the response-only solution. The function $e(m_0,\mu)$ is strictly
decreasing in the active region, so its root is unique.
If both investments are positive, $b=L/(2\theta s)$ gives
$s^*=L/(2\kappa\theta)$, and $\eta=e$ gives
\eqref{eq:joint-active-solution}.

It remains to check deterrence. Any optimal deterrent must be the
least-cost one. At $m\mu=\theta$, the marginal benefit vector is
\[
\tau\left(\frac{L}{2\theta},\frac{L}{2\mu^2}\right),
\qquad 0\leq\tau\leq1.
\]
If $m_D>m_0$, complementarity requires
$\tau=2\kappa\theta/L$, so deterrence is optimal exactly when
$\kappa\leq L/(2\theta)$; the repair condition follows from
least-cost allocation. If $m_D=m_0$, then
$\mu_D=\theta/m_0>\delta$ and complementarity requires
$\tau=2\eta\theta^2/(Lm_0^2)$, giving the second threshold.
The respective least-cost cases satisfy
$\eta>\kappa m_0^2/\theta$ and
$\eta\leq\kappa m_0^2/\theta$. These inequalities make the two
case-specific criteria equivalent to the union in
\eqref{eq:joint-deterrence-region}.
\end{proof}

The result distinguishes affordability from technical feasibility.
With unbounded investment and positive effectiveness, deterrence
is feasible; the thresholds identify when buying it is optimal.
If neither threshold holds, the defender may still fund both
capabilities while accepting some compromise risk.

There is also a useful distinction within this last regime.
Equation~\eqref{eq:joint-active-solution} gives
$p^*=1-s^*$, which does not depend on $\eta$ while both
investments remain positive. Cheaper repair then increases
$\mu^*$ and reduces $m^*$, keeping their product fixed.
Compromise probability stays unchanged, but expected compromise
time $W^*=p^*/\mu^*$ falls. Better response capability can
therefore improve the security outcome even when the frequency
of compromise does not change.

\section{Numerical Experiments}
\label{sec:strategic-experiments}

We now study how the equilibrium changes with defense and attack costs, then simulate the discovery and repair process at selected equilibria. The cost sweeps solve deterministic optimization problems. The final experiment samples the random waiting times underlying the model.

We use normalized units with baseline rates $m_0=\delta=1$,
attacker productivity $r=1$, attacker value $v=4$, and defender
loss $L=8$. Except when varied explicitly, $c_A=1$, so
$\theta=4$. Costs are per modeled opportunity and rates are per
normalized time unit. These parameters are illustrative rather
than calibrated to empirical attack data.

For each parameter setting, we compute the optimum by evaluating
the baseline, the optimum on each baseline edge, the least-cost
deterrent, and the feasible interior solution from
Proposition~\ref{prop:joint-regimes}. The response-only edge
requires a one-dimensional root solve; the other candidates are
explicit. We select the candidate with the smallest total objective.

We compare this calculation with constrained numerical optimization on 300 randomized parameter settings. We draw each of $\theta,L,\kappa,\eta,m_0,\delta$ independently and uniformly on a logarithmic scale over $[0.1,10]$. The numerical formulation introduces a nonnegative loss variable and uses two starting allocations. All settings pass the marginal
optimality checks and the deterrence criterion. The largest objective discrepancy, divided by $1+|\Phi^*|$, is below $3\times10^{-11}$.

\subsection{Experiment 1: investment regimes}

We vary $\kappa$ and $\eta$ independently over $[0.125,16]$, giving 25,921 equilibria in total. Figure~\ref{fig:cost-regimes} shows the optimal regime and the resulting expected compromise time, where we see all five predicted regimes. At the baseline, the marginal benefits are $b_0=2$ and $e_0=6$. Accordingly, no additional investment is optimal when $\kappa\geq2$ and $\eta\geq6$. The deterrence condition simplifies to $ \kappa\leq1\quad\text{or}\quad\eta\leq0.25$. Thus either sufficiently inexpensive discovery or sufficiently inexpensive repair makes deterrence optimal.

Between the baseline and deterrence regions, the defender invests in one or both capabilities while attack remains profitable. For example, $(\kappa,\eta)=(1.5,2)$ gives $(m^*,\mu^*)=(1.089,1.633)$, compromise probability $1/3$, and expected compromise time $0.204$. Positive investment in both
capabilities therefore need not imply deterrence. Across this grid, expected compromise time ranges from zero to the baseline value $0.5$ time units per opportunity.

\begin{figure}[t]
\centering
\includegraphics[width=.5\columnwidth]{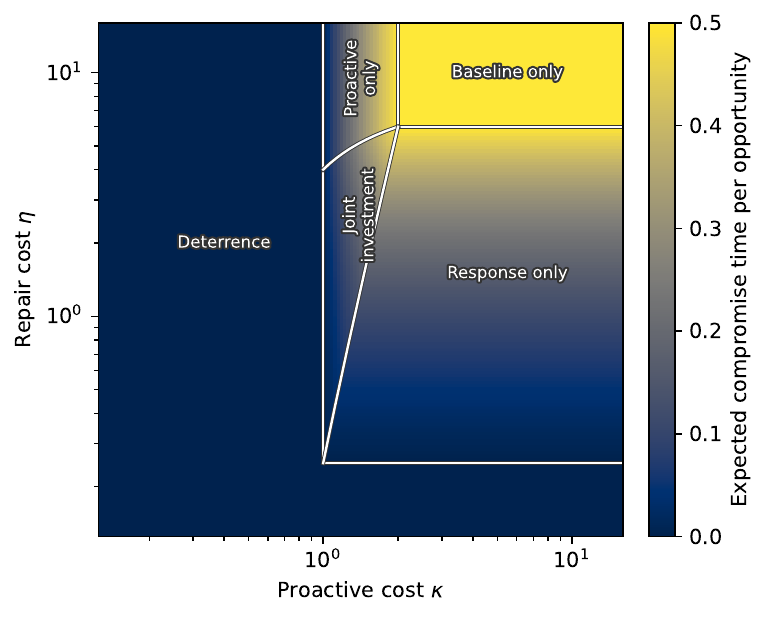}
\caption{Equilibrium outcomes as proactive and repair costs vary, with $m_0=\delta=r=c_A=1$, $v=4$, and $L=8$. The heatmap shows expected compromise time per opportunity, and white lines indicate investment regimes. Deterred equilibria are grouped together.}
\label{fig:cost-regimes}
\end{figure}

\subsection{Experiment 2: effect of attacker cost on defensive investment}

Next, we fix $\kappa=0.6$ and $\eta=1.2$ and vary $c_A$ over $[0.125,8]$. We compare joint optimization with a benchmark that fixes $\mu=\delta=1$ and optimizes discovery alone.
Figure~\ref{fig:attack-cost-sweep} reports the defensive rates, attacker effort, expected compromise time, and defender objective.

When attacks are inexpensive, joint optimization places substantial
weight on shortening their duration. At $c_A=0.125$, the optimum
has $m^*=1$ and $\mu^*=2.40$: only response receives additional
investment. Compared to fixed repair, expected compromise time goes down ($0.79$ vs $0.30$), as does total defender cost ($6.57$ vs $4.10$). This comparison accounts for the extra response investment as well as the reduction in expected loss.

At $c_A=0.25$, both capabilities receive investment. The joint solution is $(m^*,\mu^*)=(1.21,2.30)$, whereas fixed repair requires $m^*=2.78$. Both solutions have compromise probability $0.58$, but their expected compromise times are $0.25$ and
$0.58$, respectively. The joint solution improves the outcome through shorter compromises, even though their probability is unchanged. This is the constant-product behavior identified after Proposition~\ref{prop:joint-regimes}.

Both models reach deterrence at $c_A=0.6$. Allowing response investment reduces the cost of deterrence rather than shifting this particular threshold. At that boundary, joint investment costs $2.58$, compared with $3.40$ under fixed repair. Once $c_A\geq4$, the baseline itself deters attack and neither model purchases additional capability. The discovery rate is
nonmonotone over the sweep: it initially increases as stronger
proactive protection becomes worthwhile, then decreases once
deterrence can be maintained with less investment.

\begin{figure*}[t]
\centering
\includegraphics[width=\textwidth]{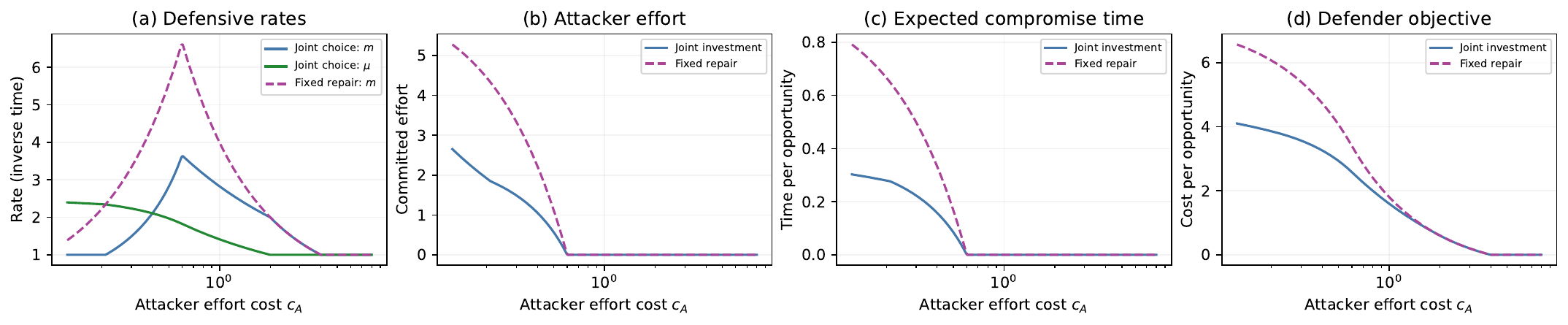}
\caption{Response to changing attacker effort cost with
$\kappa=0.6$ and $\eta=1.2$. Joint investment chooses both rates;
the benchmark fixes $\mu=1$ and optimizes discovery. All remaining
parameters have their baseline values. The two models deter at $c_A=0.6$ in this experiment, but can have substantially different costs and compromise durations.}
\label{fig:attack-cost-sweep}
\end{figure*}

Figure~\ref{fig:stochastic-validation} shows close agreement between
the simulated outcomes and the analytical predictions, both for
compromise probability and for mean compromise time. The duration
tails also follow the predicted survival function,
$$  \Prob(T>t)=p^*e^{-\mu^*t}, \qquad t\geq0. $$
We report Wilson score intervals for compromise probability and
approximate normal intervals for mean duration, both at the
95\% confidence level.

The survival curves distinguish the two ways investment reduces
exposure. Proactive-only investment lowers the probability that
compromise begins, but leaves its conditional duration unchanged:
the tail starts lower and decays at the baseline rate. Response
investment increases the decay rate, reducing the chance that a
compromise persists for a long time. At equilibrium, faster repair
also reduces the attacker's incentive to invest, so its benefit
extends to the probability of compromise. Joint investment combines
these effects.

These results explain why mean compromise time should be read
alongside the duration distribution. A lower mean can reflect
fewer compromises, shorter compromises, or both; the survival
curves reveal which mechanism produces the improvement. In
particular, faster repair limits prolonged exposure even in
regimes where attack remains profitable.

\begin{table}[t]
\centering
\begin{minipage}[t]{0.4\columnwidth}
\vspace{0pt}
\centering
\small
\setlength{\tabcolsep}{3pt}
\begin{tabular}{@{}lrrl@{}}
\hline
\multicolumn{4}{c}{\textbf{Compromise probability}} \\
Scenario & Theory & Simulation & 95\% interval \\
\hline
Baseline & 0.5000 & 0.4979 & [0.4957, 0.5000] \\
Proactive & 0.3333 & 0.3331 & [0.3310, 0.3351] \\
Response & 0.3581 & 0.3586 & [0.3565, 0.3607] \\
Joint & 0.3333 & 0.3344 & [0.3323, 0.3364] \\
Deterrence & 0.0000 & 0.0000 & [0.0000, 0.0000] \\
\hline
\multicolumn{4}{c}{\textbf{Mean compromise time}} \\
Scenario & Theory & Simulation & 95\% interval \\
\hline
Baseline & 0.5000 & 0.4974 & [0.4936, 0.5012] \\
Proactive & 0.3333 & 0.3343 & [0.3310, 0.3376] \\
Response & 0.2173 & 0.2189 & [0.2169, 0.2210] \\
Joint & 0.2041 & 0.2043 & [0.2023, 0.2063] \\
Deterrence & 0.0000 & 0.0000 & [0.0000, 0.0000] \\
\hline
\end{tabular}

\end{minipage}
\hspace{5em}
\begin{minipage}[t]{0.4\columnwidth}
\vspace{0pt}
\centering
\includegraphics[width=\linewidth]{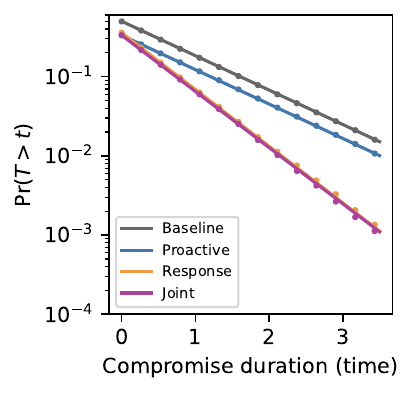}
\end{minipage}
\caption{(Left) theoretical and simulated compromise probabilities and mean compromise times with 95\% intervals. Durations are in normalized time units per opportunity. (Right) analytical survival probabilities (lines) and empirical values (dots) for compromise duration.
}
\label{fig:stochastic-validation}
\end{table}

\section{Towards Verifiable Feedback and Retention}
\label{sec:discussion}

The sufficient conditions developed in this paper assume the perfect scenario when all detected vulnerabilities are true positives that get repaired, and earlier repairs remain effective. Applying these results to a deployed system therefore requires a way to establish that its testing and update procedures satisfy those conditions. A detector reporting a failure does not by itself establish a vulnerability, and an update described
as a repair does not establish that the vulnerability has been removed. Likewise, keeping a record of earlier repairs does not ensure that their protection survives subsequent updates. These distinctions motivate a \textit{verification layer} connecting the evidence produced by testing to the claims made about the evolving defense.


A useful feedback record should explain the failure, the specific policy it violates, and the desired behaviour. For an interaction with an AI system, this requires preserving the
relevant input and conversation 
context, the observed response, the system and defense versions, and the criterion used to judge the outcome. The detector's decision is evidence to be assessed against that criterion. This assessment can be automated for verifiable requirements, whereas contextual judgments about harmfulness may require additional adjudication. In either case, the requirement must be explicit enough to distinguish
a genuine repair from a change in how failures are labelled.

Resolution of a correctly identified vulnerability should also be verifiable. For a randomized system, one successful retest is insufficient to establish reliable protection. A defensive claim must specify an acceptable failure probability and the evidence
supporting it.

Retention requires carrying these obligations forward. An update that fixes the latest failure can still invalidate an earlier repair, so checking only newly discovered attacks does not establish the monotonicity condition $D_t\subseteq D_{t+1}$. A practical starting point is to maintain a growing collection of accepted adversarial cases together with benign cases that constrain unnecessary blocking. Let $A_t$ and $B_t$ denote these collections at round $t$, and let $c_{t+1}(x)=1$ mean that a proposed input filter
blocks prompt $x$. A finite verification contract is
\[
  c_{t+1}(x)=1 \quad \forall x\in A_t,
  \qquad
  c_{t+1}(x)=0 \quad \forall x\in B_t.
\]
Newly accepted findings extend $A_t$, while earlier cases remain part of the contract. Every subsequent update must therefore satisfy both the new repair obligation and the retained obligations. The benign cases ensure that blocking all inputs cannot satisfy
the contract. Finally, extending such guarantees to a family of attacks would require a specification and an argument covering that family.

As a proof of concept, we implement a verification layer in Lean on top of a number of attacks in garak \cite{garak} (se to illustrate how such obligations can be represented and checked in an automated manner\footnote{\url{https://github.com/rajlakshmichavan/garak-in-lean}}.

\section{Conclusion}

In this paper, we investigate the feasibility of vulnerability discovery mechanisms and the investment a defender needs to make for these mechanisms to effectively defend an AI system in deployment. 
%
Building on this framework requires studying the full feedback loop in deployed systems. Longitudinal evaluations should track whether discoveries lead to effective repairs, whether those repairs generalize, and how often later updates reintroduce earlier failures.
Our verification layer implementation provides a starting point for making these claims checkable. Extending it to versioned repair histories and independently validated test cases would  help connect the assumptions of the theory to evidence from practice.

A further direction is to make investment decisions part of the evolving feedback process. Defenders learn about attack costs, discovery effectiveness, and repair reliability as testing proceeds, while attackers adapt to the resulting protections. A repeated game
with uncertain capabilities, imperfect retention, and constraints on benign-task performance would allow investment to respond to that evidence. Such an extension could guide when to expand testing, when to prioritize repair, and when existing protection needs to be reassessed.

\bibliographystyle{abbrvnat}
\bibliography{reference}

@inproceedings{rainbowteaming,
 author = {Samvelyan, Mikayel and Raparthy, Sharath Chandra and Lupu, Andrei and Hambro, Eric and Markosyan, Aram H. and Bhatt, Manish and Mao, Yuning and Jiang, Minqi and Parker-Holder, Jack and Foerster, Jakob and Rockt\"{a}schel, Tim and Raileanu, Roberta},
 booktitle = {Advances in Neural Information Processing Systems},
 doi = {10.52202/079017-2229},
 editor = {A. Globerson and L. Mackey and D. Belgrave and A. Fan and U. Paquet and J. Tomczak and C. Zhang},
 pages = {69747--69786},
 publisher = {Curran Associates, Inc.},
 title = {Rainbow Teaming: Open-Ended Generation of Diverse Adversarial Prompts},
 volume = {37},
 year = {2024}
}

@techreport{vassilev2025taxonomy,
  author = {Vassilev, Apostol and Oprea, Alina and Fordyce, Alie
            and Anderson, Hyrum and Davies, Xander and Hamin, Maia},
  title = {Adversarial Machine Learning: A Taxonomy and Terminology
           of Attacks and Mitigations},
  institution = {National Institute of Standards and Technology},
  number = {NIST AI 100-2e2025},
  year = {2025},
  doi = {10.6028/NIST.AI.100-2e2025}
}

@inproceedings{zhan2024injecagent,
  author = {Zhan, Qiusi and Liang, Zhixiang and Ying, Zifan
            and Kang, Daniel},
  title = {{InjecAgent}: Benchmarking Indirect Prompt Injections
           in Tool-Integrated Large Language Model Agents},
  booktitle = {Findings of the Association for Computational
               Linguistics: ACL 2024},
  pages = {10471--10506},
  publisher = {Association for Computational Linguistics},
  year = {2024},
  doi = {10.18653/v1/2024.findings-acl.624},
}

@techreport{souppaya2022patch,
  author = {Souppaya, Murugiah and Scarfone, Karen},
  title = {Guide to Enterprise Patch Management Planning:
           Preventive Maintenance for Technology},
  institution = {National Institute of Standards and Technology},
  number = {NIST SP 800-40 Rev. 4},
  year = {2022},
  doi = {10.6028/NIST.SP.800-40r4}
}

@misc{pyrit,
      title={PyRIT: A Framework for Security Risk Identification and Red Teaming in Generative AI System}, 
      author={Gary D. Lopez Munoz and Amanda J. Minnich and Roman Lutz and Richard Lundeen and Raja Sekhar Rao Dheekonda and Nina Chikanov and Bolor-Erdene Jagdagdorj and Martin Pouliot and Shiven Chawla and Whitney Maxwell and Blake Bullwinkel and Katherine Pratt and Joris de Gruyter and Charlotte Siska and Pete Bryan and Tori Westerhoff and Chang Kawaguchi and Christian Seifert and Ram Shankar Siva Kumar and Yonatan Zunger},
      year={2024},
      eprint={2410.02828},
      archivePrefix={arXiv},
      primaryClass={cs.CR},
      url={https://arxiv.org/abs/2410.02828}, 
}

@misc{garak,
      title={garak: A Framework for Security Probing Large Language Models}, 
      author={Leon Derczynski and Erick Galinkin and Jeffrey Martin and Subho Majumdar and Nanna Inie},
      year={2024},
      eprint={2406.11036},
      archivePrefix={arXiv},
      primaryClass={cs.CL},
      url={https://arxiv.org/abs/2406.11036}, 
}

@article{vassilev2026robust,
  author  = {Apostol Vassilev},
  title   = {Robust {AI} Security and Alignment:
             A Sisyphean Endeavor?},
  journal = {IEEE Security \& Privacy},
  year    = {2026},
  volume  = {24},
  number  = {3},
  pages   = {52--58},
  doi     = {10.1109/MSEC.2026.3678214}
}

@inproceedings{perez2022red,
  title = {Red Teaming Language Models with Language Models},
  author = {Perez, Ethan and Huang, Saffron and Song, Francis and Cai, Trevor and Ring, Roman and Aslanides, John and Glaese, Amelia and McAleese, Nat and Irving, Geoffrey},
  booktitle = {Proceedings of the 2022 Conference on Empirical Methods in Natural Language Processing},
  year = {2022},
  pages = {3419--3448},
  publisher = {Association for Computational Linguistics},
  doi = {10.18653/v1/2022.emnlp-main.225},
}

@inproceedings{hong2024curiosity,
  title = {Curiosity-driven Red-teaming for Large Language Models},
  author = {Hong, Zhang-Wei and Shenfeld, Idan and Wang, Tsun-Hsuan and Chuang, Yung-Sung and Pareja, Aldo and Glass, James and Srivastava, Akash and Agrawal, Pulkit},
  booktitle = {International Conference on Learning Representations},
  year = {2024},
}

@inproceedings{ge2024mart,
  title = {{MART}: Improving {LLM} Safety with Multi-round Automatic Red-Teaming},
  author = {Ge, Suyu and Zhou, Chunting and Hou, Rui and Khabsa, Madian and Wang, Yi-Chia and Wang, Qifan and Han, Jiawei and Mao, Yuning},
  booktitle = {Proceedings of the 2024 Conference of the North American Chapter of the Association for Computational Linguistics: Human Language Technologies (Volume 1: Long Papers)},
  year = {2024},
  pages = {1927--1937},
  publisher = {Association for Computational Linguistics},
  doi = {10.18653/v1/2024.naacl-long.107},
}

@inproceedings{mazeika2024harmbench,
  title = {{HarmBench}: A Standardized Evaluation Framework for Automated Red Teaming and Robust Refusal},
  author = {Mazeika, Mantas and Phan, Long and Yin, Xuwang and Zou, Andy and Wang, Zifan and Mu, Norman and Sakhaee, Elham and Li, Nathaniel and Basart, Steven and Li, Bo and Forsyth, David and Hendrycks, Dan},
  booktitle = {Proceedings of the 41st International Conference on Machine Learning},
  series = {Proceedings of Machine Learning Research},
  volume = {235},
  pages = {35181--35224},
  year = {2024},
  publisher = {PMLR},
}

@inproceedings{qi2025safety,
  title = {Safety Alignment Should Be Made More Than Just a Few Tokens Deep},
  author = {Qi, Xiangyu and Panda, Ashwinee and Lyu, Kaifeng and Ma, Xiao and Roy, Subhrajit and Beirami, Ahmad and Mittal, Prateek and Henderson, Peter},
  booktitle = {International Conference on Learning Representations},
  year = {2025},
}

@inproceedings{majumdar2025red,
  title = {Red Teaming {AI} Red Teaming},
  author = {Majumdar, Subhabrata and Pendleton, Brian and Gupta, Abhishek},
  booktitle = {Proceedings of the 2025 Conference on Applied Machine Learning for Information Security},
  series = {Proceedings of Machine Learning Research},
  volume = {299},
  pages = {66--86},
  publisher = {PMLR},
  year = {2025},
}

@inproceedings{anil2024many,
  title = {Many-shot Jailbreaking},
  author = {Anil, Cem and Durmus, Esin and Panickssery, Nina and others},
  booktitle = {Advances in Neural Information Processing Systems},
  volume = {37},
  year = {2024},
}

@inproceedings{kumar2024certifying,
  title = {Certifying {LLM} Safety against Adversarial Prompting},
  author = {Kumar, Aounon and Agarwal, Chirag and Srinivas, Suraj and Li, Aaron Jiaxun and Feizi, Soheil and Lakkaraju, Himabindu},
  booktitle = {Conference on Language Modeling},
  year = {2024},
}

@article{gordon2002economics,
  title = {The Economics of Information Security Investment},
  author = {Gordon, Lawrence A. and Loeb, Martin P.},
  journal = {ACM Transactions on Information and System Security},
  volume = {5},
  number = {4},
  pages = {438--457},
  year = {2002},
  doi = {10.1145/581271.581274},
}

@article{ebel2024economics,
  title = {Economics and Optimal Investment Policies of Attackers and Defenders in Cybersecurity},
  author = {Ebel, Austin and Mitra, Debasis},
  journal = {Journal of Cybersecurity},
  volume = {10},
  number = {1},
  pages = {tyae019},
  year = {2024},
  doi = {10.1093/cybsec/tyae019},
}

@article{korzhyk2011stackelberg,
  title = {{Stackelberg} vs. {Nash} in Security Games: An Extended Investigation of Interchangeability, Equivalence, and Uniqueness},
  author = {Korzhyk, Dmytro and Yin, Zhengyu and Kiekintveld, Christopher and Conitzer, Vincent and Tambe, Milind},
  journal = {Journal of Artificial Intelligence Research},
  volume = {41},
  pages = {297--327},
  year = {2011},
  doi = {10.1613/jair.3269},
}

@misc{han2026dynamic,
  title = {Toward a Dynamic {Stackelberg} Game-Theoretic Framework for Agentic {AI} Defense Against {LLM} Jailbreaking},
  author = {Han, Zhengye and Zhu, Quanyan},
  year = {2026},
  eprint = {2507.08207},
  archivePrefix = {arXiv},
  primaryClass = {cs.AI},
  note = {Version 2, revised March 2, 2026; accepted to the ICLR 2026 AIMS Workshop},
  url = {https://arxiv.org/abs/2507.08207v2}
}

@article{russinovich2025crescendo,
  title = {Great, Now Write an Article About That: The Crescendo Multi-Turn LLM Jailbreak Attack},
  author = {Russinovich, Mark and Salem, Ahmed and Eldan, Ronen},
  journal = {arXiv preprint arXiv:2404.01833},
  year = {2025},
}

@article{mehrotra2024tap,
  title = {Tree of Attacks: Jailbreaking Black-Box LLMs with Crafted Prompts},
  author = {Mehrotra, Anay and Zampetakis, Manolis and Kassianik, Paul and Nelson, Blaine and Anderson, Hyrum and Singer, Yaron and Karbasi, Amin},
  journal = {arXiv preprint arXiv:2312.02119},
  year = {2024},
}

@misc{vijil-mbert,
  title = {Model Card for Vijil Prompt Injection},
  author = {Vijil},
  year = {2025},
  url = {https://huggingface.co/vijil/mbert-prompt-injection},
  note = {Accessed: Apr 30, 2026},
}

@inproceedings{ayub2024embeddingbasedclassifiersdetectprompt,
  title = {Embedding-based classifiers can detect prompt injection attacks},
  author = {Md. Ahsan Ayub and Subhabrata Majumdar},
  booktitle = {Conference for Applied Machine Learning for Information Security},
  year = {2024},
}

\end{document}